\documentclass[journal,draftcls,onecolumn,12pt,twoside]{IEEEtran}
\normalsize
\usepackage{cite}
\usepackage{amsmath,amssymb,amsfonts,amsthm}
\usepackage{ifthen}
\usepackage{tikz}
\usepackage{cite}
\usepackage{url}
\usepackage{titlesec}
\usepackage{subfig}
\usepackage{blkarray}
\usepackage{dsfont}
\usepackage[mathscr]{euscript}
\usepackage{enumitem}
\usepackage{algorithm}
\usepackage[noend]{algpseudocode}
\usepackage{blkarray}
\usepackage{stfloats}%
\newtheorem{theorem}{Theorem}[section]
\newtheorem{corollary}{Corollary}[section]

\theoremstyle{remark}
\newtheorem*{remark}{Remark}
\theoremstyle{definition}
\newtheorem{definition}{Definition}[section]
\newtheorem{example}{Example}[section]
\usepackage{url}
\usepackage{arydshln}
\usepackage{mathtools}
\usepackage{dsfont}

\definecolor{brickred}{cmyk}{0,0.89,0.94,0.28}
\definecolor{goldenrod}{cmyk}{0,0.10,0.84,0}
\definecolor{purple}{cmyk}{0.45,0.86,0,0}
\definecolor{rawsienna}{cmyk}{0,0.72,1,0.45}
\definecolor{olivegreen}{cmyk}{0.64,0,0.95,0.40}
\definecolor{peach}{cmyk}{0,0.5,0.7,0}
\definecolor{darkolive}{rgb}{0.,0.4,0.}
\colorlet{grey}{gray!40}

\def\0{{\mathbf 0}}
\def\1{{\mathbf 1}}

\usepackage[cmintegrals]{newtxmath}
\newlist{mylist}{enumerate}{3}
\setlist[mylist]{label={}}
\usepackage{graphicx}
\usepackage{textcomp}
\def\BibTeX{{\rm B\kern-.05em{\sc i\kern-.025em b}\kern-.08em
		T\kern-.1667em\lower.7ex\hbox{E}\kern-.125emX}}

\ifCLASSINFOpdf
\else
\fi

\begin{document}
%
\title{Sampling-Based Estimates of the Sizes of Constrained Subcodes of Reed-Muller Codes}
%
%
%

\author{
	\IEEEauthorblockN{V.~Arvind~Rameshwar,~\IEEEmembership{Graduate Student Member,~IEEE,}
		Shreyas Jain,
		and Navin~Kashyap,~\IEEEmembership{Senior Member,~IEEE}}
	
%
\thanks{V.~A.~Rameshwar and N.~Kashyap are with the Department
of Electrical Communication Engineering, Indian Institute of Science, Bengaluru 560012, India (email: vrameshwar@iisc.ac.in; nkashyap@iisc.ac.in). S.~Jain is with the Deparment of Mathematical Sciences, IISER Mohali 140306 (email: ms20098@iisermohali.ac.in).}
\thanks{This work was supported in part by a Qualcomm Innovation Fellowship India 2022. The work of V.~A.~Rameshwar was supported by a Prime Minister’s Research Fellowship from the Ministry of Education,
Government of India.}}

%
%

\markboth{}%
{Rameshwar, Jain, and Kashyap: Sampling-Based Size Estimates of Constrained Subcodes of RM Codes}
%



\maketitle

\begin{abstract}
This paper develops an algorithmic approach for obtaining approximate, numerical estimates of the sizes of subcodes of Reed-Muller (RM) codes, all of the codewords in which satisfy a given constraint. Our algorithm is based on a statistical physics technique for estimating the partition functions of spin systems, which in turn makes use of a sampler that produces RM codewords according to a Gibbs distribution. The Gibbs distribution is designed so that it is biased towards codewords that respect the constraint. We apply our method to approximately compute the sizes of runlength limited (RLL) subcodes and obtain estimates of the weight distribution of moderate-blocklength RM codes. We observe that the estimates returned by our method are close to the true sizes when these sizes are either known or computable by brute-force search; in other cases, our computations provide provably robust estimates. As an illustration of our methods, we provide estimates of the weight distribution of the RM$(9,4)$ code.
\end{abstract}


%
\IEEEpeerreviewmaketitle

\section{Introduction}
%
%
%
%
\IEEEPARstart{R}{eed}-Muller (RM) codes are a family of binary linear codes that are obtained by the evaluations of Boolean polynomials on the points of the Boolean hypercube. These algebraic codes have been of interest to practitioners for several decades, beginning with their adoption in deep-space communication \cite{deepspace} to their present consideration in the 5G cellular standardization process, wherein it was argued (see, e.g., \cite{5g1,5g2}) that short-blocklength RM codes perform at least as well as modified
polar codes, low-density parity-check (LDPC) codes, and tail-biting convolutional codes, under maximum-likelihood (ML) decoding. 

Recent breakthrough theoretical progress \cite{kud1,kud3,Reeves,abbesandon} has also shown that RM codes are in fact capacity-achieving for general binary-input memoryless symmetric (BMS) channels, under both bitwise maximum a-posteriori probability (bit-MAP) and blockwise maximum a-posteriori probability (block-MAP) decoding. However, these works do not consider  the setting where the memoryless communication or storage channel requires that its input codewords respect a certain constraint, which may stem from the physical limitations of the medium over which transmission or storage occurs (see \cite{Roth} for more on constrained systems and coding). Very little, in fact, is known about the design of good codes over such input-constrained stochastic noise channels.

In this context, recent work by the authors \cite{arnk23tit} has explored constrained coding schemes derived from RM codes, for use over specific runlength limited (RLL) input-constrained BMS channels. The key idea in that work was the design of coding schemes using \emph{constrained subcodes} of RM codes over the input-constrained BMS channel; since the parent RM code is capacity-achieving over the (unconstrained) BMS channel, we have that constrained subcodes also have error probabilities going to zero as the blocklength goes to infinity, under optimal (ML) decoding. While the work \cite{arnk23tit} computed bounds on the rates of RLL constrained subcodes of RM codes, a natural question that arises is: what are the true sizes (or rates) of constrained subcodes of RM codes? Obtaining insights into this question will help us understand the rates achievable using RM codes over input-constrained BMS channels. This paper seeks to address this question by deriving numerical estimates of the true sizes, via sampling techniques. Our technique makes use of a well-known statistical physics approach for estimating the partition functions of spin systems; such an approximate counting technique has been largely unexplored in the coding theory literature, and we believe that there is scope for its broader application to other problems of interest in coding theory.

In this paper, we work with RLL and constant-weight constraints and seek to obtain estimates of the sizes of such constrained subcodes of RM codes. In particular, via estimates of the sizes of constant-weight subcodes of RM codes, we arrive at estimates of the weight distribution of moderate-blocklength RM codes, thereby making progress on a wide-open research problem. In particular, we provide estimates of the weight distribution of the RM$(9,4)$ code, the exact weight distribution of which is still unknown. We believe that our work will augment the existing literature on analytical bounds on the weight distribution of RM codes (see \cite{kaufman,sberlo,sam,anuprao} and the survey \cite{rm_survey} for more details) and on exact computations of weight distributions of selected RM codes (see \cite{sloaneberlekamp,sugino,kasamitokura,kasami2.5d}).

\subsection{Related Work}

There exists some prior work on investigating the sizes of constrained subcodes of linear codes and their variants (such as cosets). In particular, the work in \cite{pvk} (see also the early work \cite{Ferreira}) gave a simple existential lower bound on the maximum rates of constrained subcodes of cosets of linear codes of rate $R\in (0,1)$. The lower bound proved on such rates is $R+\kappa-1$, where $\kappa$ is the noiseless capacity (see \cite[Chapter 3]{Roth}) of the constraint. Furthermore, using random linear coding arguments, it is possible to show the existence of linear codes of rate $R$, the rates of whose constrained \emph{subcodes} is at least the same lower bound. However, to the best of our knowledge, our prior works \cite{arnk23tit} and \cite{arnk23jsait} (see Sections III-A and IV in that paper) were the first to provide explicit analytical values of, or algorithms for computing, bounds on the sizes (or rates) of constrained subcodes of general linear codes.

In the context of the constant-weight constraint, computing the sizes of constrained subcodes of linear codes is equivalent to computing the weight distribution of the code, and there exist several results on the weight distributions of well-known linear codes (see the book \cite{mws} for more details). For  RM codes in particular, as mentioned earlier, the works \cite{sloaneberlekamp,sugino,kasamitokura,kasami2.5d} provided either analytical expressions or exact values for the weight enumerators of selected RM codes, for selected weights. More recently, the works \cite{kaufman,sberlo,sam,anuprao} and the survey \cite{rm_survey} provide bounds on the weight distribution, using techniques drawn primarily from the analysis of Boolean functions on the hypercube.

Sampling-based methods, along the lines of that explored in this work, have recently been used to design randomized decoders for linear codes (see \cite{alankrita} and references therein). Our sampling technique, in particular, uses minimum-weight codewords of Reed-Muller codes; applications of such codewords for constructing ``good'' parity-check matrices for decoding can be found in \cite{minwt}. We also mention the closely related work \cite{gamal} wherein simulated annealing, which uses Metropolis sampling steps, is used as a heuristic for identifying good short-blocklength error-correcting codes via computer search.

\subsection{A Description of Our Contribution}

Our contribution is a sampling-based technique for obtaining approximate estimates of the sizes of constrained subcodes of RM codes. While we focus on runlength limited (RLL) constraints and the constant-weight constraint in this paper, we believe that our techniques can easily be extended to, and will give good results for, other constraints of interest, too. We first develop a simple Metropolis algorithm (which is a special instance of the Markov chain Monte Carlo (MCMC) method) for drawing samples from (a distribution close to) a suitably-defined Gibbs (or Boltzmann) distribution on the codewords of the RM code. Importantly, our sampler involves a ``nearest-neighbour'' proposal distribution, which uses minimum-weight codewords of RM codes. Our sampler, for sufficiently large ``inverse temperature'' parameters, can produce samples from exponentially-small (compared to the size of the parent RM code) constrained subcodes of RM codes, and is hence of independent interest. This sampling method is then bootstrapped to a well-known procedure for estimating partition functions of spin systems in statistical physics, which helps us arrive at our required size estimates. We numerically verify the soundness of our algorithm by comparing the values it returns with the true values, when the true values can be computed by a direct brute-force search. We also provide theoretical guarantees of the sample complexity of our estimation procedure (for a fixed error in approximation), and demonstrate that the number of samples, and hence the time taken to run the estimation algorithm, is only polynomial in the blocklength of the RM code.

\section{Preliminaries and Notation}
We use $\mathbb{F}_2$ to denote the binary field, i.e., the set $\{0,1\}$ equipped with modulo-$2$ arithmetic. 
We use bold letters such as $\mathbf{x}$, $\mathbf{y}$ to denote finite-length binary sequences (or vectors); the set of all finite-length binary sequences is denoted by $\{0,1\}^\star$. 
Further, when $\mathbf{x}, \mathbf{y} \in \mathbb{F}_2^n$, we denote by $\mathbf{x}+\mathbf{y}$ the vector resulting from component-wise modulo-$2$ addition. We also use the notation $w_H(\mathbf{x})$ to denote the Hamming weight of $\mathbf{x}\in \mathbb{F}^n$, which is the number of nonzero coordinates in $\mathbf{x}$. We define the indicator function of a set $\mathcal{A}$ as $\mathds{1}_\mathcal{A}$, with $\mathds{1}_\mathcal{A}(\mathbf{x}) = 1$, when $\mathbf{x}\in \mathcal{A}$, and $0$, otherwise. We use exp$(x)$ to denote $e^x$, for $x\in \mathbb{R}$. For sequences $(a_n)_{n\geq 1}$ and $(b_n)_{n\geq 1}$ of positive reals, we say that $a_n = O(b_n)$, if there exists $n_0\in \mathbb{N}$ and a positive real $M$, such that $a_n\leq M\cdot b_n$, for all $n\geq n_0$. We say that $a_n = \Theta(b_n)$, if there exist positive reals $M_1, M_2$ such that $M_1\cdot b_n \leq a_n\leq M_2\cdot b_n$ for all sufficiently large $n$. 


\subsection{Linear Codes and Constraints}
In this paper, we work with the field $\mathbb{F}_2$. We recall the following definitions of block codes and linear codes over $\mathbb{F}_2$ and their rates (see, for example, Chapter 1 of \cite{roth_coding_theory}).

\begin{definition}
	An $(n,M)$ block code $\mathcal{C}$ over $\mathbb{F}_2$ is a nonempty subset of $\mathbb{F}_2^n$, with $|\mathcal{C}| = M$. The rate of the block code $\mathcal{C}$ is given by
	$
	\text{rate}(\mathcal{C}) := \frac{\log_2 M}{n}.
	$
\end{definition}
{
	\begin{definition}
		An $[n,k]$ linear code $\mathcal{C}$ over $\mathbb{F}_2$ is an $(n,2^k)$ block code that is a subspace of $\mathbb{F}_2^n$.
	\end{definition}
}

For details on constrained coding, we refer the reader to \cite{Roth}. In this paper, we consider only two constraints: the $(d,\infty)$-runlength limited (RLL) constraint and the constant-weight constraint.
\begin{definition}
	The $(d,\infty)$-RLL constraint admits only those binary sequences $\mathbf{x}\in \{0,1\}^\star$ that have at least $d$ $0$s between successive $1$s. We let $S^{(d)}$ denote the set of $(d,\infty)$-RLL constrained sequences.
\end{definition}
\begin{definition}
	The constant-weight constraint represented by the set $W^{(\omega)}$ of constrained sequences admits only those binary sequences $\mathbf{x}\in \{0,1\}^\star$ of Hamming weight exactly $\omega$.
\end{definition}
\subsection{Reed-Muller Codes}
\label{sec:rmintro}
We now recall the definition of the binary Reed-Muller (RM) family of codes and some of their basic facts that are relevant to this work. Codewords of binary RM codes consist of the evaluation vectors of multivariate polynomials over the binary field $\mathbb{F}_2$. Consider the polynomial ring $\mathbb{F}_2[x_1,x_2,\ldots,x_m]$ in $m$ variables. Note that any polynomial $f\in \mathbb{F}_2[x_1,x_2,\ldots,x_m]$ can be expressed as the sum of \emph{monomials} of the form $\prod_{j\in S:S\subseteq [m]} x_j$, since $x^2 = x$ over the field $\mathbb{F}_2$. For a polynomial $f\in \mathbb{F}_2[x_1,x_2,\ldots,x_m]$ and a binary vector $\mathbf{z} = (z_1,\ldots,z_m)\in \mathbb{F}_2^m$, we write {$f(\mathbf{z})=f(z_1,\ldots,z_m)$} as the evaluation of $f$ at $\mathbf{z}$. The evaluation points are ordered according to the standard lexicographic order on strings in $\mathbb{F}_2^m$, i.e., if $\mathbf{z} = (z_1,\ldots,z_m)$ and $\mathbf{z}^{\prime} = (z_1^{\prime},\ldots,z_m^{\prime})$ are two distinct evaluation points, then, $\mathbf{z}$ occurs before $\mathbf{z}^{\prime}$ iff for some $i\geq 1$, we have $z_j = z_j^{\prime}$ for all $j<i$, and $z_i < z_i^{\prime}$. Now, let Eval$(f):=\left({f(\mathbf{z})}:\mathbf{z}\in \mathbb{F}_2^m\right)$ be the evaluation vector of $f$, where the coordinates $\mathbf{z}$ are ordered according to the standard lexicographic order. 

\begin{definition}[see Chap. 13 in \cite{mws}, or \cite{rm_survey}]
	{For $0\leq r\leq m$}, the $r^{\text{th}}$-order binary Reed-Muller code RM$(m,r)$ is defined as
	\[
	\text{RM}(m,r):=\{\text{Eval}(f): f\in \mathbb{F}_2[x_1,x_2,\ldots,x_m],\ \text{deg}(f)\leq r\},
	\]
	where $\text{deg}(f)$ is the degree of the largest monomial in $f$, and the degree of a monomial $\prod_{j\in S: S\subseteq [m]} x_j$ is simply $|S|$. 
\end{definition}

It is known that the evaluation vectors of all the distinct monomials in the variables $x_1,\ldots, x_m$ are linearly independent over $\mathbb{F}_2$. Hence, RM$(m,r)$ has dimension $\binom{m}{\le r} := \sum_{i=0}^{r}{m \choose i}$. Moreover, RM$(m,r)$ is a $\left[2^m,\binom{m}{\le r}\right]$ linear code. It is also known that RM$(m,r)$ has minimum Hamming distance ${d}_{\text{min}}(\text{RM}(m,r))=2^{m-r}$. Furthermore (see \cite[Thm. 8, Chap. 13, p. 380]{mws}) we have that each minimum-weight codeword of RM$(m,r)$ is the characteristic vector of an $(m-r)$-dimensional affine subspace of $\mathbb{F}_2^m$, i.e., the vector having ones in those coordinates $\mathbf{z}\in \{0,1\}^m$ that lie in the affine subspace. In particular, in this paper, we use this fact to efficiently sample  a uniformly random minimum-weight codeword of a given RM code. Another fact of use to us is that the collection of minimum-weight codewords spans RM$(m,r)$, for any $m\geq 1$ and $r\leq m$ (see \cite[Thm. 12, Chap. 13, p. 385]{mws}).
\section{Sampling-Based Estimates for Constrained Subcodes of RM Codes}
\label{sec:estimate}
In this section, we discuss our sampling-based procedure for computing (approximate) size estimates for constrained subcodes of RM codes. We are primarily interested in RLL and constant-weight constraints.

Suppose that we are given a constraint represented by a set $\mathcal{A}\subseteq \{0,1\}^n$ of constrained sequences, and an $[n,k]$ linear code $\mathcal{C}$ \footnote{We mention that while the technique for approximating $Z$ described here applies to \emph{any} code $\mathcal{C}$ (not necessarily linear), we restrict our attention to RM codes in this paper, since in this family of codes, it is possible to generate samples from the distribution $p_\beta$ (see \eqref{eq:pb}) in a computationally efficient manner, as will be required later (see Section \ref{sec:sampler}).}. We are interested in obtaining estimates of the quantity
\begin{equation}
	\label{eq:Z}
	Z = \sum_{\mathbf{c}\in \mathcal{C}} \mathds{1}\{\mathbf{c}\in \mathcal{A}\} = |\mathcal{C}\cap \mathcal{A}|.
\end{equation}

Clearly, when the dimension $k$ of $\mathcal{C}$ is large, a direct calculation of the above quantity is computationally intractable, as it involves roughly $2^k$ additions. The work in \cite{arnk23jsait} transformed this counting problem to one in the space of the dual code $\mathcal{C}^\perp$; however, even if the Fourier transform (over $\mathbb{F}_2^n$) of $\mathds{1}_\mathcal{A}$ were easy to compute, we would still have to perform $2^{n-k}$ additions. Our aim therefore is to compute good approximations to $Z$.

To this end, consider the following probability distribution supported on the codewords of $\mathcal{C}$:
\begin{equation}
	\label{eq:pb}
p_\beta(\mathbf{x}) = \frac{1}{Z_\beta}\cdot e^{-\beta\cdot E(\mathbf{x})}\cdot \mathds{1}_\mathcal{C}(\mathbf{x}), \quad \mathbf{x}\in \{0,1\}^n,
\end{equation}
where $\beta>0$ is some fixed real number (in statistical physics, $\beta$ is termed as ``inverse temperature''), and $E: \{0,1\}^n\to [0,\infty)$ is an ``energy function'' such that $E(\mathbf{x}) = 0$ if $\mathbf{x}\in \mathcal{A}$ and is strictly positive, otherwise. Note also that the ``partition function'' or the normalization constant
\begin{equation}
	\label{eq:Zb}
	Z_\beta = \sum_{\mathbf{c}\in \mathcal{C}} e^{-\beta\cdot E(\mathbf{x})}.
\end{equation}
In the limit as $\beta\to \infty$, it can be argued that the distribution $p_\beta$ becomes the uniform distribution over the ``ground states'' or zero-energy vectors in $\mathcal{C}$ (see, e.g., \cite[Chapter 2]{mezard}); more precisely,
\begin{equation}
	\label{eq:limit}
\lim_{\beta\to \infty} p_\beta(\mathbf{x}) = \frac{1}{Z}\cdot \mathds{1}_{E_0}(\mathbf{x}),
\end{equation}
where $\mathcal{C}_0:= \{\mathbf{c}\in \mathcal{C}:\ E(\mathbf{c}) = 0\}$. Clearly, by the definition of the energy function, $\mathcal{C}_0=\mathcal{C}\cap \mathcal{A}$, with $Z = |\mathcal{C}_0| = \lim_{\beta\to \infty}Z_\beta$. Our intention now is to compute an approximation to $Z$ via this perspective. While we have defined the energy function $E$ with arbitrary constraints in mind, we first illustrate examples of $E$ for the two constraints of interest in this paper.
\begin{example}
	\label{eg:rll}
	For the $(d,\infty)$-RLL constraint $S^{(d)}$, we define the energy function $E = E^{(d)}$ to be the number of indices where there is a violation of the constraint, i.e., for $\mathbf{x}\in \{0,1\}^n$,
	\[
	E^{(d)}(\mathbf{x}) = \#\{j:\ (x_j,x_{j+i})=(1,1), \text{ for }1\leq j\leq n-d\text{ and some }1\leq i\leq d\}.
	\]
	Clearly, $E^{(d)}$ equals $0$ at codewords that lie in $S^{(d)}$, and is strictly positive, otherwise. 
	We let $p_\beta^{(d)}$ be the probability distribution induced on codewords of $\mathcal{C}$ using the energy function $E^{(d)}$ as in \eqref{eq:pb}.
\end{example}
\begin{example}
	\label{eg:wt}
	For the constant-weight constraint $W^{(\omega)}$, with $0\leq \omega\leq n$, we define the energy function $E = E^{(\omega)}$ as
	\[
	E^{(\omega)}(\mathbf{x}) = |w_H(\mathbf{x})-\omega|.
	\]
	Again, $E^{(\omega)}$ equals $0$ at codewords of weight exactly $\omega$, and is strictly positive, otherwise. We let $p_\beta^{(\omega)}$ be the probability distribution induced on codewords of $\mathcal{C}$ using the energy function $E^{(\omega)}$ as in \eqref{eq:pb}.
\end{example}
In this work, following \eqref{eq:limit}, we shall use the partition function $Z_{\beta^\star}$, where $\beta^\star$ is suitably large, as an estimate of $Z$ (the question of how large $\beta^\star$ must be for $Z_{\beta^\star}$ to be a good approximation to $Z$ is taken up in Section \ref{sec:theory}). In the section that follows, we shall outline a fairly standard method from the statistical physics literature \cite{statphy} (see also Lecture 4 in \cite{sinclair}) to compute $Z_{\beta^\star}$ corresponding to the two constraints defined above, when $\beta^\star$ is large. 
\subsection{An Algorithm for Computing $Z_{\beta^\star}$}
\label{sec:practice}
Recall that our intention is to compute $Z_{\beta^\star}$, for the two kinds of probability distributions $p_\beta^{(d)}$ and $p_\beta^{(\omega)}$ (see Examples \ref{eg:rll} and \ref{eg:wt}), where $\beta^\star$ is large. We propose a statistical physics-based procedure for computing an estimate of $Z_{\beta^\star}$. We provide a qualitative description of the procedure in this section; the exact values of the parameters required to guarantee a close-enough estimate will be provided in Section \ref{sec:theory}. 

The key idea in this method is to express $Z_{\beta^\star}$ as a telescoping product of ratios of partition functions, for smaller values of $\beta$. We define a sequence (or a ``cooling schedule'')
\begin{equation}
	\label{eq:betas}
0 = \beta_0< \beta_1< \ldots < \beta_\ell= \beta^\star,
\end{equation}
where $\beta_i = \beta_{i-1}+\frac{1}{n}$ and $\ell$ is a large positive integer, and write 
\begin{equation}
	\label{eq:zbstar}
	Z_{\beta^\star} = Z_{\beta_0}\times\prod_{i=1}^{\ell}\frac{Z_{\beta_i}}{Z_{\beta_{i-1}}}.
\end{equation}
Observe from \eqref{eq:Zb} that $Z_{\beta_0} = Z_0 = |\mathcal{C}|=2^k$, where $k$ is the dimension of $\mathcal{C}$. Furthermore, for $1\leq i\leq\ell$, we have from \eqref{eq:pb} that
\begin{align}
	\frac{Z_{\beta_i}}{Z_{\beta_{i-1}}} &=  \frac{1}{Z_{\beta_{i-1}}}\sum_{\mathbf{c}\in \mathcal{C}} \text{exp}(-\beta_iE(\mathbf{c})) \notag\\
	&= \frac{1}{Z_{\beta_{i-1}}}\sum_{\mathbf{c}\in \mathcal{C}}\text{exp}(-\beta_{i-1}E(\mathbf{c}))\cdot \text{exp}((\beta_{i-1}-\beta_i)E(\mathbf{c})) \notag\\
	&= \mathbb{E}[\text{exp}(-E(\mathbf{c})/n)] \label{eq:expectation},
\end{align}
where the expectation is over codewords $\mathbf{c}$ drawn according to $p_{\beta_{i-1}}$. In other words, the ratio $\frac{Z_{\beta_i}}{Z_{\beta_{i-1}}}$ can be computed as the expected value of a random variable $X_i:=\text{exp}(-E(\mathbf{c})/n)$, where $\mathbf{c}$ is drawn according to $p_{\beta_{i-1}}$. A description of how such a random variable can be sampled is provided in Algorithm \ref{alg:sample} in Section \ref{sec:sampler}. We now explain how this perspective is used to obtain an estimate of $Z_{\beta^\star}$ in \eqref{eq:zbstar}. For every $i$, for large $t$, we sample i.i.d. random variables $X_{i,j}$, $1\leq j\leq t$, which have the same distribution as $X_i$. We ensure that the $X_{i,j}$s are independent across $i$ as well. We then estimate the expected value in \eqref{eq:expectation} by a sample average, i.e., we define the random variable
\begin{equation}
	\label{eq:sampleav}
Y_i:= \frac{1}{t}\sum_{j=1}^t X_{i,j}.
\end{equation}
Finally, the estimate for $Z_{\beta^\star}$ (see \eqref{eq:zbstar}) that we shall use is
\begin{equation}
	\label{eq:estimate}
\widehat{Z}_{\beta^\star} = Z_{\beta_0}\times \prod_{i=1}^\ell Y_i.
\end{equation}
Note that we then have, by independence of the $X_{i,j}$s and hence of the $Y_i$s, that $\mathbb{E}[\widehat{Z}_{\beta^\star}] = Z_0\times \prod_{i=1}^\ell \mathbb{E}[Y_i] = Z_{\beta^\star}$. A summary of our algorithm is shown as Algorithm \ref{alg:estimate}. 

We now return to the question of sampling a codeword from $\mathcal{C}$, efficiently, according to the distribution $p_\beta$, for any $\beta>0$. It is here that we shall use the fact that $\mathcal{C}$ is chosen to be a Reed-Muller code RM$(m,r)$, for $m\geq 1$ and $r\leq m$.
\begin{algorithm}[t]
	\caption{Estimating $Z$ via $Z_{\beta^\star}$}
	\label{alg:estimate}
	\begin{algorithmic}[1]	
		\Procedure{Estimator}{$\beta^\star$}
		\State Fix a cooling schedule $0=\beta_0<\beta_1<\ldots<\beta_\ell = \beta^\star$, as in \eqref{eq:betas}.
		\State Fix a large $t\in \mathbb{N}$.
		\For{$i=1:\ell$}
		\State Use Algorithm \ref{alg:sample} to generate $t$ i.i.d. samples $\mathbf{c}_{i,1},\ldots,\mathbf{c}_{i,t}$.
		\State For $1\leq j\leq t$, set $X_{i,j} \leftarrow \text{exp}((\beta_{i-1}-\beta_i)E(\mathbf{c}_{i,j}))$.
		\State Compute $Y_i = \frac{1}{t}\sum_{j=1}^t X_{i,j}$.
		\EndFor
		\State Output $\widehat{Z}_{\beta^\star} = |\mathcal{C}|\times \prod_{i=1}^\ell Y_i$.
		\EndProcedure	
	\end{algorithmic}
\end{algorithm} 
\subsection{An Algorithm for Sampling RM Codewords According to $p_\beta$}
\label{sec:sampler}
Our approach to generating samples from the distribution $p_\beta$, when $\mathcal{C}$ is a Reed-Muller code is a simple ``nearest-neighbour'' Metropolis algorithm, which is a special instance of Monte Carlo Markov Chain (MCMC) methods---for more details on the Metropolis algorithm and on general MCMC methods, see Chapter 3 in \cite{mcmcbook}.

The idea behind the Metropolis algorithm is to induce a Markov chain on the space of codewords of $\mathcal{C}$, such that the stationary distribution of the Markov chain equals $p_\beta$. Let $\Delta$ be the collection of minimum-weight codewords in $\mathcal{C}$. Consider the following ``symmetric proposal distribution'' $\{P(\mathbf{c}_1,\mathbf{c}_2):\ \mathbf{c}_1,\mathbf{c}_2\in \mathcal{C}\}$, where $P(\mathbf{c}_1,\mathbf{c}_2)$ is the conditional probability of ``proposing'' codeword $\mathbf{c}_2$ given that we are at codeword $\mathbf{c}_1$:
\begin{equation}
	\label{eq:proposal}
P(\mathbf{c}_1, \mathbf{c}_2) = 
\begin{cases}
	\frac{1}{|\Delta|},\ \text{if $\mathbf{c}_2 = \mathbf{c}_1+\overline{\mathbf{c}}$, for some $\overline{\mathbf{c}}\in \Delta$},\\
	0,\ \text{otherwise}.
\end{cases}
\end{equation}
Clearly, $P$ is symmetric in that $P(\mathbf{c}_1, \mathbf{c}_2) = P(\mathbf{c}_2, \mathbf{c}_1)$, for all $\mathbf{c}_1, \mathbf{c}_2\in \mathcal{C}$. Our Metropolis algorithm begins at a randomly initialized codeword. When the algorithm is at codeword $\mathbf{c}_1$, it ``accepts'' the proposal of codeword $\mathbf{c}_2$ with probability $\min\left(1,\frac{p_\beta(\mathbf{c}_2)}{p_\beta(\mathbf{c}_1)}\right)$, and moves to $\mathbf{c}_2$. Now, observe that since $\mathcal{C}$ is a Reed-Muller code, it is easy to sample a codeword $\mathbf{c}_2$ that differs from $\mathbf{c}_1$ by a minimum-weight codeword; in other words one can efficiently sample a uniformly random minimum-weight codeword $\overline{\mathbf{c}}$. More precisely, we know (see Section \ref{sec:rmintro}) that minimum-weight codewords are characteristic vectors of $(m-r)$-dimensional affine subspaces of $\mathbb{F}_2^m$. Hence, our task boils down to choosing a uniformly random $(m-r)$-dimensional affine subspace $H$. To this end, we first sample a uniformly random $(m-r)\times m$ full-rank (over $\mathbb{F}_2$) $0$-$1$ matrix $A$ and a uniformly random vector $\mathbf{b}\in \{0,1\}^m$. We then construct $H$ as
\[
H = \{\mathbf{z}: \mathbf{z}=\mathbf{x}\cdot A+\mathbf{b},\text{ for some $\mathbf{x}\in \mathbb{F}_2^{m-r}$}\}.
\] 
Note that the full-rank matrix $A$ can be sampled by first constructing a random $(m-r)\times m$ matrix $M$ whose entries are picked i.i.d. Ber$(1/2)$, and then using a rejection sampling procedure (see, for example, Appendix B.5 in \cite{mcmcbook}) that ``accepts'' only when the matrix is full-rank. Furthermore, it is well-known that the probability that the uniformly random matrix $M$ is full rank (over $\mathbb{F}_2$) obeys
\[
\Pr[M\text{ is full rank}] = \prod_{i=1}^{m-r} (1-2^{-i})\geq \prod_{i=1}^{\infty} (1-2^{-i})\approx 0.289,
\]
where the rightmost approximate equality uses sequence A048651 in \cite{oeis}. This then implies that the expected number of rejection sampling steps required before we obtain our full-rank matrix $A$ is bounded above by a constant.
 
Returning to our Metropolis algorithm, we have that our \emph{Metropolis chain} has as transition probability matrix the $|\mathcal{C}|\times|\mathcal{C}|$ matrix $Q$, where
\begin{equation}
	\label{eq:metropolis}
Q(\mathbf{c}_1,\mathbf{c}_2) = 
\begin{cases}
	P(\mathbf{c}_1, \mathbf{c}_2) \cdot \min\left(1,\frac{p_\beta(\mathbf{c}_2)}{p_\beta(\mathbf{c}_1)}\right),\ \text{if $\mathbf{c}_1\neq \mathbf{c}_2$},\\
	1-\sum_{\mathbf{c}\neq \mathbf{c}_1}P(\mathbf{c}_1,\mathbf{c})\cdot \min\left(1,\frac{p_\beta(\mathbf{c})}{p_\beta(\mathbf{c}_1)}\right),\ \text{otherwise}.
\end{cases}
\end{equation}
It can be checked that $p_\beta$ is indeed a stationary distribution of this chain. Further, suppose that $\mathbf{c}^{(\tau)}$ is the (random) codeword that this chain is at, at time $\tau\in \mathbb{N}$. Then, it is well-known that if the Metropolis chain is irreducible and aperiodic (and hence ergodic), then the distribution of $\mathbf{c}_\tau$ is close, in total variational distance, to the stationary distribution $p_\beta$ (see, for example, Theorem 4.9 in \cite{mcmcbook}), for large enough $\tau$.\footnote{An important question that can be asked is how large $\tau$ must be to guarantee that the distribution of $\mathbf{c}_\tau$ is sufficiently close to $p_\beta$. In this work, we do not address this question, but simply set $\tau$ to be large enough so that the Metropolis chain reaches the ``zero-energy'' constrained codewords within $\tau$ steps, in practice, starting from an arbitrary initial codeword.}

For our chain in \eqref{eq:metropolis}, since the set of minimum-weight codewords $\Delta$ spans $\mathcal{C}$, we have that the chain is irreducible. Further, for select constraints such as the constant-weight constraint, it can be argued that in some special cases, there always exists a pair of codewords $(\mathbf{c}_1,\mathbf{c}_2)$ such that  $\mathbf{c}_2 = \mathbf{c}_1+\overline{\mathbf{c}}$, for some $\overline{\mathbf{c}}\in \Delta$, with $p_\beta(\mathbf{c}_2)<p_\beta(\mathbf{c}_1)$. We then get that $Q(\mathbf{c}_1,\mathbf{c}_1)>0$, assuring us of aperiodicity, and hence of ergodicity, of our chain. However, for the purposes of this work, we do not concern ourselves with explicitly proving aperiodicity, and instead seek to test the soundness of our technique, numerically. Our algorithm for sampling codewords (approximately) from $p_\beta$ is given as Algorithm \ref{alg:sample}.
 \begin{algorithm}[t]
 	\caption{Sampling RM codewords approximately from $p_\beta$}
 	\label{alg:sample}
 	\begin{algorithmic}[1]	
 		\Procedure{Metropolis-Sampler}{$\mathbf{c}^{(0)}$, $\beta$, $E$}       
 		\State Initialize the Metropolis chain at the arbitrary (fixed) codeword $\mathbf{c}^{(0)}$.
 		\State Fix a large $\tau\in \mathbb{N}$.
 		\For{$i=1:\tau$}
 		\State Generate a uniformly random $(m-r)\times m$ full-rank $0$-$1$ matrix $A$ and a uniformly random vector $\mathbf{b}\in \{0,1\}^n$.
 		\State Construct $H = \{\mathbf{z}: \mathbf{z}=\mathbf{x}\cdot A+\mathbf{b},\text{ for some $\mathbf{x}\in \mathbb{F}_2^{m-r}$}\}$.
 		\State Set $\overline{\mathbf{c}}$ to be the characteristic vector of $H$ and set $\mathbf{c}\leftarrow \mathbf{c}^{(i-1)}+\overline{\mathbf{c}}$.
 		\State Set $\mathbf{c}^{(i)} \leftarrow \mathbf{c}$ with probability $\min\left(1,\text{exp}(-\beta(E(\mathbf{c})-E(\mathbf{c}^{(i-1)})))\right)$; else set $\mathbf{c}^{(i)} \leftarrow \mathbf{c}^{(i-1)}$.
 		\EndFor
 		\State Output $\mathbf{c}_\tau$.
 		\EndProcedure	
 	\end{algorithmic}
 \end{algorithm} 
\subsection{Theoretical Guarantees}
\label{sec:theory}
To demonstrate how good our estimate $\widehat{Z}_{\beta^\star}$ is, we invoke the following (well-known) theorem of Dyer and Frieze \cite{dyerfrieze} (see also Theorem 2.1 in \cite{stefankovic}):
\begin{theorem}
	\label{thm:dyer}
	Fix an $\epsilon \geq 0$. Let $U_1,\ldots,U_\ell$ be independent random variables with $\mathbb{E}[U_i^2]/(\mathbb{E}[U_i])^2 \leq B$, for some $B\geq 0$ and for $1\leq i\leq \ell$. Set $\widehat{U} = \prod_{i=1}^\ell U_i$. Also, for $1\leq i\leq \ell$, let $V_i$ be the average of $16B\ell/\epsilon^2$ independent random samples having the same distribution as $U_i$; set $\widehat{V} =  \prod_{i=1}^\ell V_i$. Then,
	\[
	\Pr\left[(1-\epsilon)\mathbb{E}[\widehat{U}]\leq \widehat{V}\leq (1+\epsilon)\mathbb{E}[\widehat{U}]\right]\geq \frac34.
	\]
\end{theorem}
As a direct corollary, we obtain the following guarantee about our estimate $\widehat{Z}_{\beta^\star}$. Let $t^\star = 16e^2\ell/\epsilon^2$.
\begin{corollary}
	\label{cor:dyer}
	Fix an $\epsilon \geq 0$. If $Y_i$ is the average of $t^\star$ i.i.d. samples having the same distribution as $X_i$ (see \eqref{eq:sampleav}), then,
	\[
	\Pr\left[(1-\epsilon)Z_{\beta^\star}\leq \widehat{Z}_{\beta^\star}\leq (1+\epsilon)Z_{\beta^\star}\right]\geq \frac34.
	\]
\end{corollary}
\begin{proof}
	Since we have that $0\leq E^{(d)}(\mathbf{c})\leq n$ and $0\leq E^{(\omega)}(\mathbf{c})\leq n$ for all $d,\omega,\text{ and }n$, it follows that the random variable $X_i \in [e^{-1},1]$ as $\beta_i-\beta_{i-1} = \frac{1}{n}$. Hence, $\mathbb{E}[X_i^2]/(\mathbb{E}[X_i])^2 \leq e^2=:B$. The proof then follows from a simple application of Theorem \ref{thm:dyer} with the observation that by the independence of the $Y_i$s, we have $\mathbb{E}[Z_{\beta_0}\times \prod_{i=1}^\ell Y_i] = Z_{\beta^\star}$, from \eqref{eq:estimate}.
\end{proof}
\begin{remark}
	Following Proposition 4.2 in \cite[Lecture 4]{sinclair}, we have that the constant on the right-hand side of Corollary \ref{cor:dyer} can be improved to $1-\gamma$ for $\gamma$ arbitrarily small, by using the new estimate $\overline{Z}_{\beta^\star}$ that is the median of $\widehat{Z}_{\beta^\star}^{(1)},\ldots,\widehat{Z}_{\beta^\star}^{(T)}$ where $T = O(\log \gamma^{-1})$ and each $\widehat{Z}_{\beta^\star}^{(i)}$, for $1\leq i\leq T$, is drawn i.i.d. according to \eqref{eq:estimate}. Hence, we obtain an estimate that lies in $[(1-\epsilon)Z_{\beta^\star},(1+\epsilon)Z_{\beta^\star}]$, for $\epsilon$ arbitrarily small, with arbitrarily high probability.
\end{remark}
Observe that the number of samples $t^\star$ required to compute a single sample average as in \eqref{eq:sampleav} is polynomial (in fact, linear) in the length $\ell$ of the cooling schedule, for a fixed $\epsilon>0$. It thus remains to specify this length $\ell$. From arguments similar to that in \cite[Lecture 4]{sinclair} (see the paragraph following Eq. (4.6) there), we have that for $\beta^\star = O(n^2)$, the value $Z_{\beta^\star}$ is such that $Z_{\beta^\star} = (1+\delta_n)Z$, for $\delta_n = \text{exp}(-\Theta(n^2))$. In other words, for large $n$, with $\beta^\star = \Theta(n^2)$, our estimate $\widehat{Z}_{\beta^\star}$ is such that
\[
\Pr[(1-\epsilon)(1+\delta_n)Z\leq \widehat{Z}_{\beta^\star}\leq (1+\epsilon)(1+\delta_n)Z]\geq \frac34,
\]
from Corollary \ref{cor:dyer}. Hence, it suffices for $\ell$ to be $\Theta(n^3)$ (since $\beta_i - \beta_{i-1} = 1/n$, for $1\leq i\leq \ell$, with $\beta_0 = 0$ and $\beta_\ell = \beta^\star$) to obtain a good estimate of the true number of constrained codewords, $Z$. We also then have that the total number of samples required, $t^\star\ell$, is $\Theta(n^6)$, for a fixed $\epsilon$, which is still only polynomial in the blocklength $n$ of $\mathcal{C}$. This must be contrasted with the number of computations required to obtain an exact value of $Z$ (by either directly going over the codewords of $\mathcal{C}$ or by using the technique in \cite{arnk23jsait}), which is at least $\min(2^k,2^{n-k})$ that is exponential in $n$ when $k$ grows linearly with $n$.
\section{Numerical Examples}
\label{sec:numerics}
In this section, we shall apply a variant of Algorithm \ref{alg:estimate} to compute estimates of the sizes of constrained subcodes of specific RM codes of small or moderate blocklengths. The method we use in this section differs from Algorithm \ref{alg:estimate} in that we do not pick a value of $\ell$ (which determines the cooling schedule completely) in advance. Instead, we shall iterate the loop in Step 4 of Algorithm \ref{alg:estimate} and keep updating the estimate $\widehat{Z}_{\beta^\star}$ until it settles to within a precribed precision $\delta\in (0,1)$. The algorithm we use in our implementations is shown as Algorithm \ref{alg:estimateprac}.

\begin{algorithm}[t]
	\caption{Estimating $Z$ via $\widehat{Z}$}
	\label{alg:estimateprac}
	\begin{algorithmic}[1]	
		\Procedure{Estimator}{}
		\State Fix a large $t\in \mathbb{N}$.
		\State Fix a (small) precision $\delta\in (0,1)$ and set $\beta \leftarrow 0$.
		\State Set \texttt{curr}$\ \leftarrow |\mathcal{C}|$ and \texttt{prev}$\ \leftarrow 0$.
		\While{$|\text{\texttt{curr}}-\text{\texttt{prev}}|>\delta$}
		\State Increment $\beta \leftarrow \beta+1/n$.
		\State Draw $t$ samples $\mathbf{c}_{1},\ldots,\mathbf{c}_{t}$ i.i.d. from $p_\beta$ using Algorithm \ref{alg:sample}.
		\State For $1\leq j\leq t$, set $X_{j} \leftarrow \text{exp}(-E(\mathbf{c}_{j})/n)$.
		\State Compute $Y = \frac{1}{t}\sum_{j=1}^t X_{j}$.
		\State Update \texttt{prev}$\ \leftarrow$ \texttt{curr} and \texttt{curr}$\ \leftarrow Y\cdot \text{\texttt{curr}}$.
		\EndWhile
		\State Output $\widehat{Z} = \text{\texttt{curr}}$.
		\EndProcedure	
	\end{algorithmic}
\end{algorithm} 

 As mentioned earlier, we shall illustrate examples for the $(d,\infty)$-RLL and constant-weight constraints. We shall also compare the estimates obtained by our method with the true values computed by either brute-force search or using other techniques in the literature.
\subsection{$(d,\infty)$-RLL Constraint}
In what follows, we shall work mainly with the settings where $d=1,2$. For the case when $d=1$, Table \ref{tab:1inftrue} catalogues the values of the estimate $\widehat{Z}$ obtained by running Algorithm \ref{alg:estimateprac}, for fixed values of the parameters $\tau$ (the number of Metropolis steps in Algorithm \ref{alg:sample}), $t$ (the number of samples drawn in every iteration in Algorithm \ref{alg:estimateprac}), and $\delta$ (the precision to which convergence of the estimate in Algorithm \ref{alg:estimateprac} is guaranteed), for different codes RM$(m,r)$. The estimates are compared with the true values $Z$ of the number of constrained codewords in RM$(m,r)$, which are computed by a brute-force search. Observe that the values of $m$ are such that the codes considered in this table are short/moderate-blocklength RM codes, for which such a brute-force enumeration is tractable\footnote{While it is in principle possible to use the Fourier-analytic techniques in \cite{arnk23jsait} to transform the counting problem for codes of rate larger than $1/2$ to a computation over codes of rate at most $1/2$, the procedure specified there still requires storing the Fourier coefficients of the indicator function of the constraint---a task that is prohibitively expensive, for even blocklengths as small as $32$ ($m=5$). We hence simply adopt the brute-force search strategy for high-rate codes, too.}. We also compare the ``rates'' of the estimates $\widehat{Z}$, given by $\frac{\log_2 \widehat{Z}}{2^m}$, with the true rates of the constrained subcodes, $\frac{\log_2 {Z}}{2^m}$. We note that the estimates $\widehat{Z}$ are close to the true values $Z$, in these cases, and closer still are the rates of the estimates to the true rates.

Next, we use our sampling-based method to compute estimates of the numbers of constrained codewords in moderate blocklength RM codes for which the brute-force enumeration strategy is computationally intractable. We then compare the rates of our estimates with results from \cite{arnk23tit}. In particular, suppose that $R=\text{rate}(\text{RM}(m,r))$. From the proofs of Theorems III.1 and III.2 in \cite{arnk23tit}, we infer that there exist $(1,\infty)$-RLL constrained subcodes of RM$(m,r)$ of rate $R^1$, where 
\[
R^1\geq
	\max\left(\frac{\log_2\binom{m-1}{\leq r-1}}{2^m},R-\frac38-\frac{1}{4\cdot 2^{m-1}}\right)=:R^1_\text{LB}.
\]
Our results are documented in Table \ref{tab:1infguess}. As can be seen from the table, the rates of our estimates are larger than the rates computed using the results of \cite{arnk23tit}, indicating that the rates analytically derived in \cite{arnk23tit} are not tight. 

Similarly, in Table \ref{tab:2inf}, comparisons are made between the estimated sizes of $(2,\infty)$-RLL constrained subcodes of RM$(m,r)$, for varying values of $m$ and $r$, with the true sizes. 

\begin{table}[t!]
	\centering
	\begin{tabular}{||c || c|| c||c||c|| c ||c|| c|| c||} 
		\hline
		$m$ & $r$ & $\tau$ & $t$ & $\delta$& $\widehat{Z}$ & $Z$ & $\frac{\log_2 \widehat{Z}}{2^m}$ & $\frac{\log_2 {Z}}{2^m}$\\
		\hline
		$4$& $2$& $5\times 10^3$& $50$& $0.1$& $80$ & $83$ & $0.3951$ & $0.3984$\\
		\hline
		$5$ & $2$ & $10^4$ & $50$ & $0.1$ & $278.446$ & $259$ & $0.2538$ & $0.2505$\\
		\hline
		$5$ & $3$ & $10^4$ & $10$ & $0.001$ & $126490$ & $89172$ & $0.5296$ & $0.5139$\\
		\hline
		$6$ & $1$ & $10^4$ & $10$ & $0.001$ & $5.551$ & $4$ & $0.0386$ & $0.0312$\\
		\hline
		$6$ & $2$ & $10^4$ & $10$ & $0.001$ & $997.7$ & $803$ & $0.1557$ & $0.1508$\\
		\hline
		$7$ & $2$ & $10^4$ & $5$ & $0.1$ & $3128.4$ & $2467$ & $0.0907$ & $0.0880$\\
		\hline
		$7$ & $2$ & $10^4$ & $100$ & $0.001$ & $2515.5$ & $2467$ & $0.0883$ & $0.0880$\\
		\hline
		$8$ & $1$ & $10^5$ & $10$ & $0.001$ & $5.3787$ & $4$ & $0.0095$ & $0.0078$\\
		\hline
	\end{tabular}
	\caption{Table of values of estimates $\widehat{Z}$ and rates $\frac{\log_2\widehat{Z}}{2^m}$ of the sizes of $(1,\infty)$-RLL constrained codewords in codes RM$(m,r)$, for different values $m\geq 1$, $r\leq m$, and for different parameters in Algorithms \ref{alg:sample} and \ref{alg:estimateprac}. The estimates are compared with the true sizes $Z$ and true rates $\frac{\log_2{Z}}{2^m}$ computed via brute-force search and enumeration.}
	\label{tab:1inftrue}
\end{table}

\begin{table}[t!]
	\centering
	\begin{tabular}{||c || c|| c||c||c|| c|| c|| c||} 
		\hline
		$m$ & $r$ & $\tau$ & $t$ & $\delta$& $\widehat{Z}$ & $\frac{\log_2\widehat{Z}}{2^m}$ & $R^1_\text{LB}$\\
		\hline
		$7$& $3$& $10^4$& $100$& $0.001$& $2.926\times 10^8$ & $0.1678$ & $0.1211$\\
		\hline
		$7$ & $4$ & $10^4$ & $100$ & $0.001$ & $1.199\times 10^{18}$ & $0.4692$ & $0.3945$\\
		\hline
		$7$ & $5$ & $10^4$ & $100$ & $0.001$ & $2.676\times 10^{24}$ & $0.6340$ & $0.5586$\\
		\hline
		$8$ & $2$ & $10^5$ & $10$ & $0.1$ & $1.255\times 10^4$ & $0.0526$ & $0.0312$\\
		\hline
		$8$ & $3$ & $10^5$ & $10$ & $0.1$ & $5.249\times 10^{10}$ & $0.1391$ & $0.1133$\\
		\hline
		$8$ & $4$ & $10^5$ & $10$ & $0.1$ & $5.754\times 10^{25}$ & $0.3343$ & $0.2598$\\
		\hline
		$8$ & $5$ & $10^5$ & $10$ & $0.1$ & $3.464\times 10^{42}$ & $0.5520$ & $0.4785$\\
		\hline
	\end{tabular}
	\caption{Table of values of estimates $\widehat{Z}$ of the numbers of $(1,\infty)$-RLL constrained codewords in moderate-blocklength codes RM$(m,r)$, for which a brute-force computation of constrained codewords is intractable.}
	\label{tab:1infguess}
\end{table}

\begin{table}[t!]
	\centering
	\begin{tabular}{||c || c|| c|| c||} 
		\hline
		$m$ & $r$ & $\widehat{Z}$ & $Z$\\
		\hline
		$4$& $1$&  $1.101$ & $1$\\
		\hline
		$4$ & $2$ &  $36.614$ & $37$\\
		\hline
		$4$ & $3$ &  $350.743$ & $303$\\
		\hline
		$5$ & $2$ &  $87.025$ & $81$\\
		\hline
		$5$ & $3$ &  $4998.2$ & $4917$\\
		\hline
		$5$ & $4$ &  $1.271\times 10^5$ & --\\
		\hline
		$6$ & $2$ &  $184.473$ & $177$\\
		\hline
		$6$ & $3$ &  $6.663\times 10^4$ & --\\
		\hline
		$7$ & $2$ &  $357.672$ & --\\
		\hline
	\end{tabular}
	\caption{Table of values of estimates $\widehat{Z}$ of the numbers of $(2,\infty)$-RLL constrained codewords in RM$(m,r)$, for varying values of $m\geq 1$ and $r\leq m$, compared with the true counts $Z$ whenever a brute-force computation is tractable; a `--' symbol is inserted in instances when it is not. Here, the parameters $\tau = 10^4$, $t=100$, and $\delta = 0.001$.}
	\label{tab:2inf}
\end{table}
\subsection{Constant-weight Constraint}
We now shift our attention to the constant-weight constraint, which allows us to get estimates of the weight enumerator $(A_{m,r}(\omega): 0\leq \omega\leq 2^m)$ of RM$(m,r)$. We use the following two facts (see, for example, the survey \cite{rm_survey}) about the weight distribution of RM codes to ease computation: (1) the weight distribution is supported only on even weights, i.e., $A_{m,r}(\omega) = 0$ if $\omega$ is odd, and (2) the weight distribution is symmetric about $\omega = 2^{m-1}$, i.e., $A_{m,r}(\omega) = A_{m,r}(n-\omega)$, for $0\leq \omega\leq 2^m$. Hence, in the sequel, we shall only compute estimates of the weight enumerators for even weights $\omega$, for $2^{m-r}\leq \omega\leq 2^{m-1}$, since the minimum distance of RM$(m,r)$ is $2^{m-r}$. Furthermore, when we compare our size or rate estimates with the true sizes $A_{m,r}(\omega)$ or rates $\frac{1}{2^m}\cdot {\log_2 A_{m,r}(\omega)}$, we shall confine our attention to only those weights $\omega$ such that $A_{m,r}(\omega)>0$. 

Figures \ref{fig:rmwt_n5r3}, \ref{fig:rmwt_n6r3}, and \ref{fig:rmwt_n7r3} show comparisons of the rates of our estimates of the weight enumerators, with the true rates. For these runs, we set $\tau = 10^4$, $t=100$, and $\delta = 0.001$. We mention that for RM$(5,3)$ and RM$(6,3)$, the true weight enumerators were obtained by brute-force computation or by using the algebraic techniques in Chapter 5 of \cite{sarwate}. For RM$(7,3)$, the true weight enumerators were obtained in \cite{sugino} by solving a system of linear equations that makes use of the MacWilliams identities, and we directly use these values. In all these cases, we observe that our estimates are close to the ground truth. 

We also use our method to estimate the weight enumerators (and their rates) of RM$(9,4)$, for selected weights. It is known (see Section 4 in \cite{sugita}) that it is possible to determine all the true weight enumerators of RM$(9,4)$ if the weight enumerators at weights $\omega = 80, 84$ are derived. To this end, we apply our method to compute estimates of the weight enumerators at these weights (see Table \ref{tab:rmwt94}), thereby allowing one to plug these estimates in the equations in \cite{sugita} obtain estimates of the entire weight distribution of RM$(9,4)$.
\begin{figure*}
	\centering
	\includegraphics[width=\textwidth]{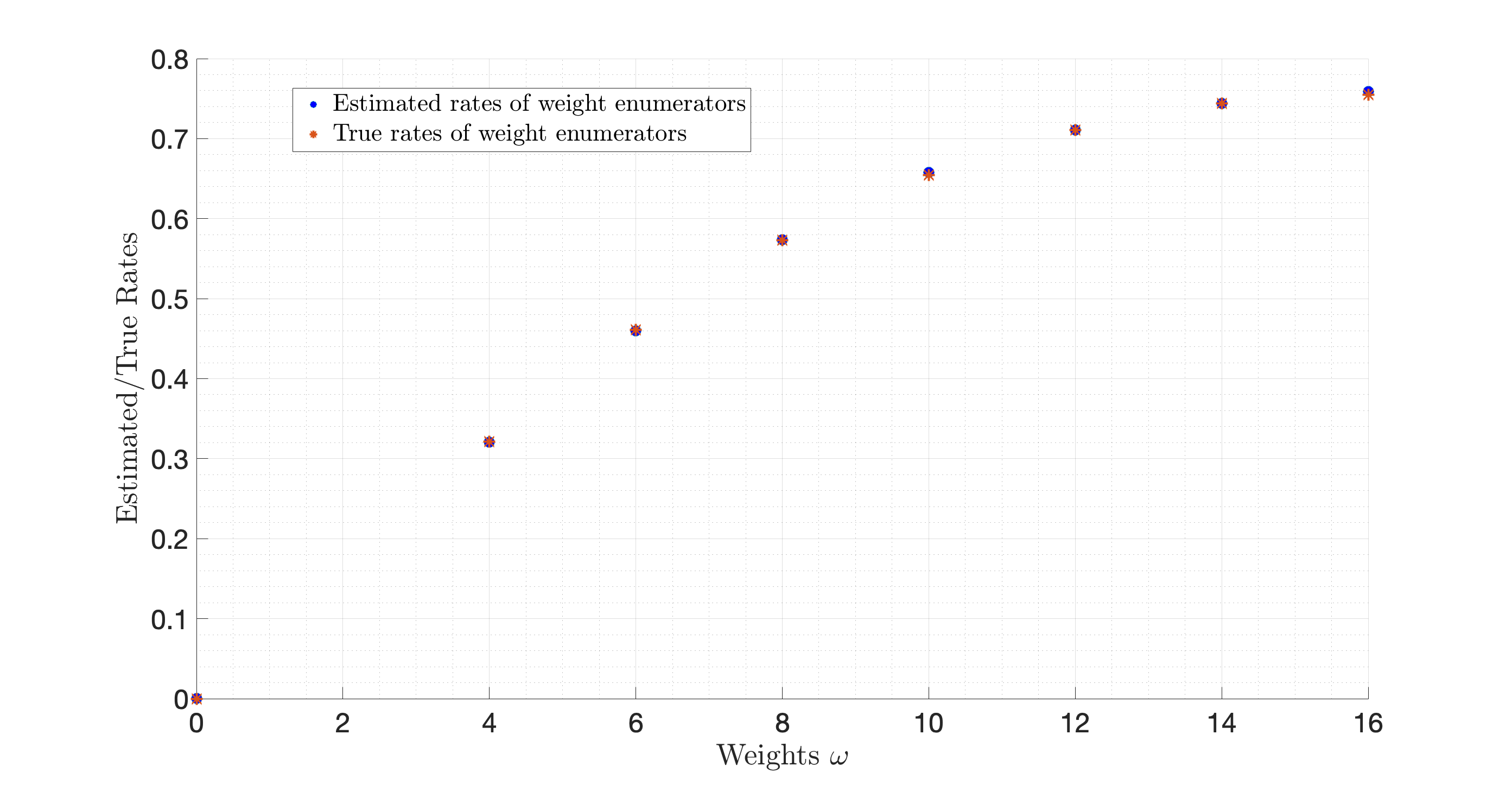}
	\caption{Plot comparing the estimates of rates of the weight enumerators obtained via our sampling-based approach, with the true rates (for weights with positive weight enumerators), for RM$(5,3)$.}
	\label{fig:rmwt_n5r3}
\end{figure*}
\begin{figure*}
	\centering
	\includegraphics[width=\textwidth]{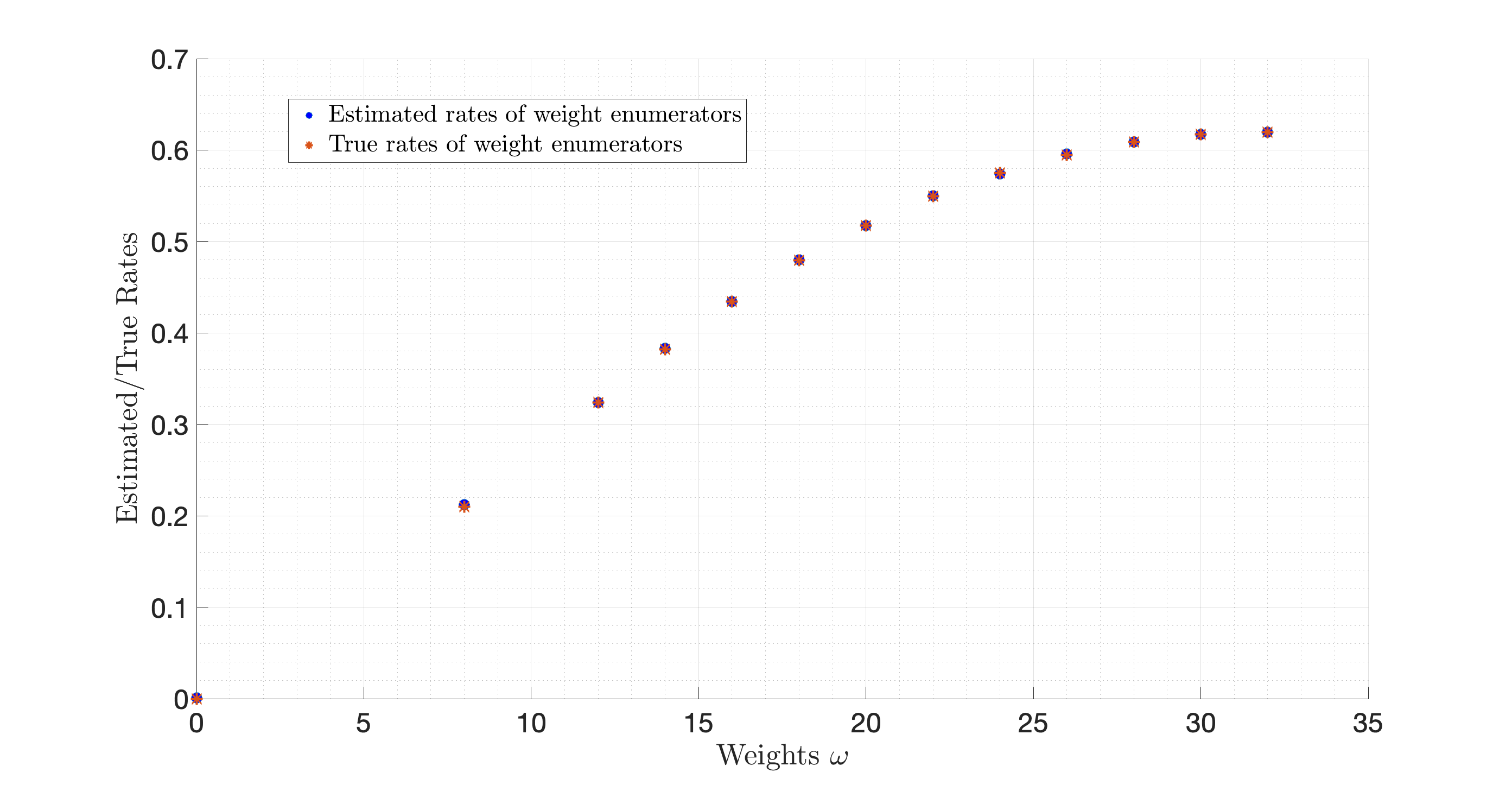}
	\caption{Plot comparing the estimates of rates of the weight enumerators obtained via our sampling-based approach with the true rates (for weights with positive weight enumerators), for RM$(6,3)$.}
	\label{fig:rmwt_n6r3}
\end{figure*}
\begin{figure*}
	\centering
	\includegraphics[width=\textwidth]{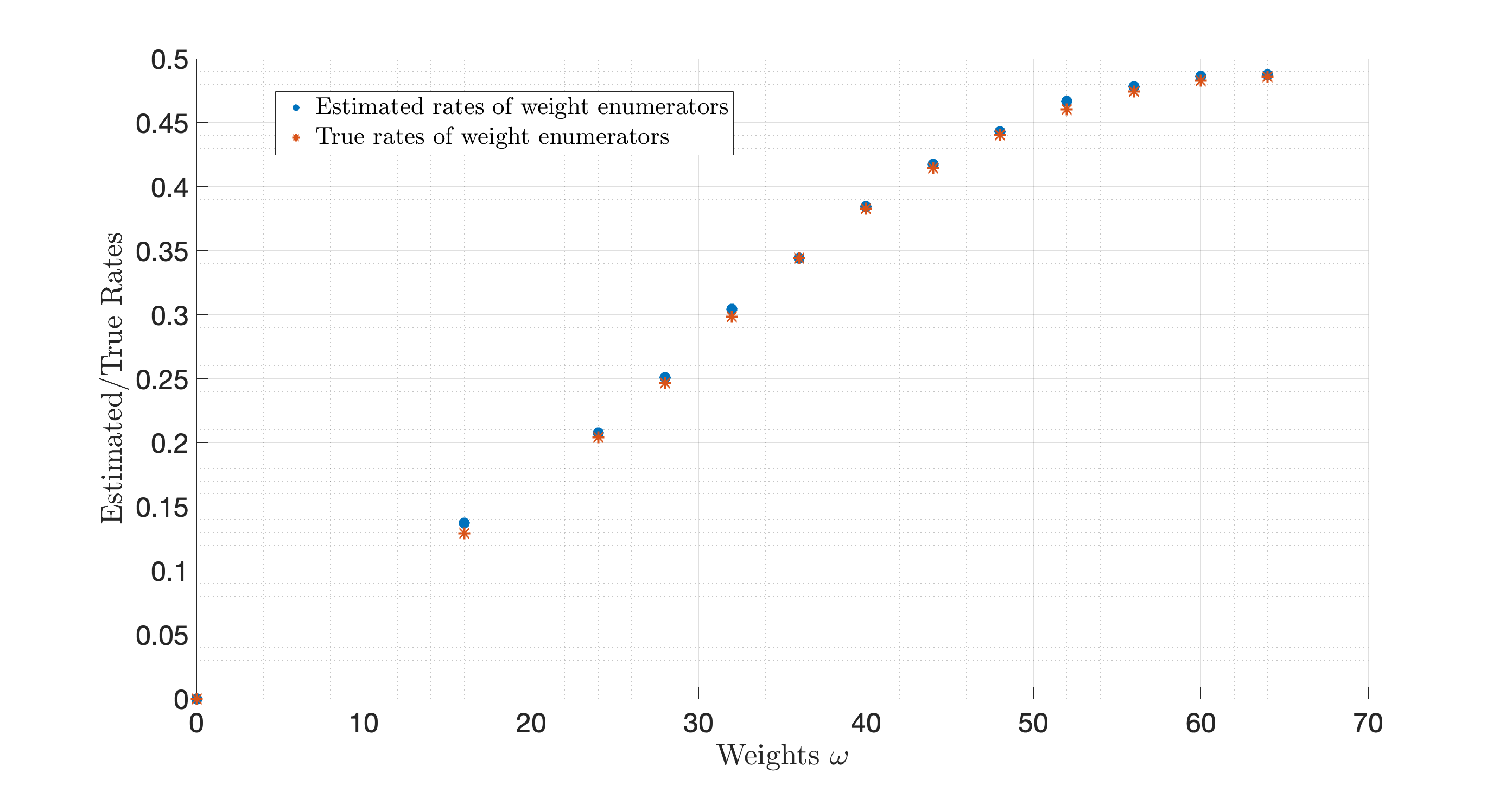}
	\caption{Plot comparing the estimates of rates of the weight enumerators obtained via our sampling-based approach with the true rates (for weights with positive weight enumerators), for RM$(7,3)$, obtained from \cite{sugino}.}
	\label{fig:rmwt_n7r3}
\end{figure*}

\begin{table}[t!]
	\centering
	\begin{tabular}{||c || c|| c||} 
		\hline
		$\omega$& $\widehat{Z}$ & $\frac{\log_2 \widehat{Z}}{2^m}$\\
		\hline
		$76$ & $4.079\times 10^{22}$ & $0.1467$\\
		\hline
		$80$ & $2.991\times 10^{23}$ & $0.1523$\\
		\hline
		$84$ & $1.429\times 10^{25}$ & $0.1632$\\
		\hline
	\end{tabular}
	\caption{Table of estimates $\widehat{Z}$ and $\frac{\log_2 \widehat{Z}}{2^m}$ of the weight enumerators and their rates, respectively, for RM$(9,4)$, for weights $\omega = 76, 80, 84$. Here, the parameters $\tau = 5\times 10^5$, $t=1$, and $\delta = 0.001$.}
	\label{tab:rmwt94}
\end{table}
\section{Conclusion}
In this paper, we proposed a novel sampling-based approach for computing estimates of the sizes or rates of $(d,\infty)$-runlength limited (RLL) and constant-weight constrained subcodes of the Reed-Muller (RM) family of codes. In particular, via our estimates of the sizes of constant-weight subcodes, we obtained estimates of the weight distribution of moderate-blocklength RM codes. We observed that our estimates are close to the true sizes (or rates) for those RM codes where a direct computation of the true values is computationally feasible. Moreover, using our techniques, we obtained estimates of the weight distribution of RM$(9,4)$, whose true weight enumerators are not known for all weights. We also provided theoretical guarantees of the robustness of our estimates and argued that for a fixed error in approximation, our proposed algorithm uses a number of samples that is only polynomial in the blocklength of the code. Such sampling-based approaches have been largely unexplored for counting problems of interest in coding theory and we believe that there is much scope for the application of such techniques to other open problems.


%



\ifCLASSOPTIONcaptionsoff
  \newpage
\fi



%
%
\bibliographystyle{IEEEtran}
{\footnotesize
	\bibliography{references.bib}}

%




\end{document}